\documentclass[pra,aps,nopacs,twocolumn,twoside,superscriptaddress]{revtex4}

\usepackage{amsmath,amsfonts,amssymb,caption,color,epsfig,graphics,graphicx,hyperref,latexsym,mathrsfs,revsymb,theorem,url,verbatim,epstopdf,mathtools,enumerate}

\usepackage{amsmath,amsfonts,amssymb,caption,hyperref,color,epsfig,graphics,graphicx,latexsym,mathrsfs,revsymb,theorem,url,verbatim,epstopdf,cleveref}
\hypersetup{colorlinks,linkcolor={blue},citecolor={blue},urlcolor={red}}

\newtheorem{definition}{Definition}
\newtheorem{proposition}[definition]{Proposition}
\newtheorem{lemma}[definition]{Lemma}

\newtheorem{theorem}[definition]{Theorem}
\newtheorem{corollary}[definition]{Corollary}
\newtheorem{conjecture}[definition]{Conjecture}

\newtheorem{remark}[definition]{Remark}
\newtheorem{example}[definition]{Example}
\newtheorem{question}[definition]{Question}
\newtheorem{memo}[definition]{Memo}

\def\squareforqed{\hbox{\rlap{$\sqcap$}$\sqcup$}}
\def\qed{\ifmmode\squareforqed\else{\unskip\nobreak\hfil
\penalty50\hskip1em\null\nobreak\hfil\squareforqed
\parfillskip=0pt\finalhyphendemerits=0\endgraf}\fi}
\def\endenv{\ifmmode\;\else{\unskip\nobreak\hfil
\penalty50\hskip1em\null\nobreak\hfil\;
\parfillskip=0pt\finalhyphendemerits=0\endgraf}\fi}
\newenvironment{proof}{\noindent \textbf{{Proof.~} }}{\qed}
\def\Dbar{\leavevmode\lower.6ex\hbox to 0pt
{\hskip-.23ex\accent"16\hss}D}
\makeatletter
\def\url@leostyle{%
  \@ifundefined{selectfont}{\def\UrlFont{\sf}}{\def\UrlFont{\small\ttfamily}}}
\makeatother
\def\bcj{\begin{conjecture}}
\def\ecj{\end{conjecture}}
\def\bcr{\begin{corollary}}
\def\ecr{\end{corollary}}
\def\bd{\begin{definition}}
\def\ed{\end{definition}}
\def\bea{\begin{eqnarray}}
\def\eea{\end{eqnarray}}
\def\bem{\begin{enumerate}}
\def\eem{\end{enumerate}}
\def\bex{\begin{example}}
\def\eex{\end{example}}
\def\bim{\begin{itemize}}
\def\eim{\end{itemize}}
\def\bl{\begin{lemma}}
\def\el{\end{lemma}}
\def\bma{\begin{bmatrix}}
\def\ema{\end{bmatrix}}
\def\bpf{\begin{proof}}
\def\epf{\end{proof}}
\def\bpp{\begin{proposition}}
\def\epp{\end{proposition}}
\def\bqu{\begin{question}}
\def\equ{\end{question}}
\def\br{\begin{remark}}
\def\er{\end{remark}}
\def\bt{\begin{theorem}}
\def\et{\end{theorem}}
\def\bmm{\begin{memo}}
\def\emm{\end{memo}}

\def\btb{\begin{tabular}}
\def\etb{\end{tabular}}

\newcommand{\nc}{\newcommand}

\def\r{\rho}

\def\ps{\psi}

\def\G{\Gamma}

 \nc{\bbA}{\mathbb{A}} \nc{\bbB}{\mathbb{B}} \nc{\bbC}{\mathbb{C}}
 \nc{\bbD}{\mathbb{D}} \nc{\bbE}{\mathbb{E}} \nc{\bbF}{\mathbb{F}}
 \nc{\bbG}{\mathbb{G}} \nc{\bbH}{\mathbb{H}} \nc{\bbI}{\mathbb{I}}
 \nc{\bbJ}{\mathbb{J}} \nc{\bbK}{\mathbb{K}} \nc{\bbL}{\mathbb{L}}
 \nc{\bbM}{\mathbb{M}} \nc{\bbN}{\mathbb{N}} \nc{\bbO}{\mathbb{O}}
 \nc{\bbP}{\mathbb{P}} \nc{\bbQ}{\mathbb{Q}} \nc{\bbR}{\mathbb{R}}
 \nc{\bbS}{\mathbb{S}} \nc{\bbT}{\mathbb{T}} \nc{\bbU}{\mathbb{U}}
 \nc{\bbV}{\mathbb{V}} \nc{\bbW}{\mathbb{W}} \nc{\bbX}{\mathbb{X}}
 \nc{\bbZ}{\mathbb{Z}}

 \nc{\bA}{{\bf A}} \nc{\bB}{{\bf B}} \nc{\bC}{{\bf C}}
 \nc{\bD}{{\bf D}} \nc{\bE}{{\bf E}} \nc{\bF}{{\bf F}}
 \nc{\bG}{{\bf G}} \nc{\bH}{{\bf H}} \nc{\bI}{{\bf I}}
 \nc{\bJ}{{\bf J}} \nc{\bK}{{\bf K}} \nc{\bL}{{\bf L}}
 \nc{\bM}{{\bf M}} \nc{\bN}{{\bf N}} \nc{\bO}{{\bf O}}
 \nc{\bP}{{\bf P}} \nc{\bQ}{{\bf Q}} \nc{\bR}{{\bf R}}
 \nc{\bS}{{\bf S}} \nc{\bT}{{\bf T}} \nc{\bU}{{\bf U}}
 \nc{\bV}{{\bf V}} \nc{\bW}{{\bf W}} \nc{\bX}{{\bf X}}
 \nc{\bZ}{{\bf Z}}

\nc{\cA}{{\cal A}} \nc{\cB}{{\cal B}} \nc{\cC}{{\cal C}}
\nc{\cD}{{\cal D}} \nc{\cE}{{\cal E}} \nc{\cF}{{\cal F}}
\nc{\cG}{{\cal G}} \nc{\cH}{{\cal H}} \nc{\cI}{{\cal I}}
\nc{\cJ}{{\cal J}} \nc{\cK}{{\cal K}} \nc{\cL}{{\cal L}}
\nc{\cM}{{\cal M}} \nc{\cN}{{\cal N}} \nc{\cO}{{\cal O}}
\nc{\cP}{{\cal P}} \nc{\cQ}{{\cal Q}} \nc{\cR}{{\cal R}}
\nc{\cS}{{\cal S}} \nc{\cT}{{\cal T}} \nc{\cU}{{\cal U}}
\nc{\cV}{{\cal V}} \nc{\cW}{{\cal W}} \nc{\cX}{{\cal X}}
\nc{\cZ}{{\cal Z}}

\nc{\hA}{{\hat{A}}} \nc{\hB}{{\hat{B}}} \nc{\hC}{{\hat{C}}}
\nc{\hD}{{\hat{D}}} \nc{\hE}{{\hat{E}}} \nc{\hF}{{\hat{F}}}
\nc{\hG}{{\hat{G}}} \nc{\hH}{{\hat{H}}} \nc{\hI}{{\hat{I}}}
\nc{\hJ}{{\hat{J}}} \nc{\hK}{{\hat{K}}} \nc{\hL}{{\hat{L}}}
\nc{\hM}{{\hat{M}}} \nc{\hN}{{\hat{N}}} \nc{\hO}{{\hat{O}}}
\nc{\hP}{{\hat{P}}} \nc{\hR}{{\hat{R}}} \nc{\hS}{{\hat{S}}}
\nc{\hT}{{\hat{T}}} \nc{\hU}{{\hat{U}}} \nc{\hV}{{\hat{V}}}
\nc{\hW}{{\hat{W}}} \nc{\hX}{{\hat{X}}} \nc{\hZ}{{\hat{Z}}}

\nc{\hn}{{\hat{n}}}

\def\dim{\mathop{\rm Dim}}

\def\max{\mathop{\rm max}}
\def\min{\mathop{\rm min}}

\def\rank{\mathop{\rm rank}}

\def\sr{\mathop{\rm sr}}

\def\tr{\mathop{\rm Tr}}

\def\dg{\dagger}

\def\ox{\otimes}

\newcommand{\bra}[1]{\langle#1|}
\newcommand{\ket}[1]{|#1\rangle}
\newcommand{\proj}[1]{| #1\rangle\!\langle #1 |}
\newcommand{\ketbra}[2]{|#1\rangle\!\langle#2|}
\newcommand{\braket}[2]{\langle#1|#2\rangle}

\newcommand{\norm}[1]{\lVert#1\rVert}

\def\Dbar{\leavevmode\lower.6ex\hbox to 0pt
{\hskip-.23ex\accent"16\hss}D}

\begin{document}

\title{Entanglement distillation in terms of a conjectured matrix inequality}

\date{\today}


\author{Yize Sun}\email[]{sunyize@buaa.edu.cn}
\affiliation{School of Mathematics and Systems Science, Beihang University, Beijing 100191, China}

\author{Lin Chen}\email[]{linchen@buaa.edu.cn (corresponding author)}
\affiliation{School of Mathematics and Systems Science, Beihang University, Beijing 100191, China}
\affiliation{International Research Institute for Multidisciplinary Science, Beihang University, Beijing 100191, China}

\begin{abstract}
Entanglement distillation is a basic task in quantum information, and the distillable entanglement of three bipartite reduced density matrices from a tripartite pure state has been studied in [Phys. Rev. A 84, 012325 (2011)]. We extend this result to tripartite mixed states by studying a conjectured matrix inequality, namely $\rank(\sum_i R_i \otimes S_i)\le K \rank(\sum_i R_i^T \otimes S_i)$ holds for any bipartite matrix $M=\sum_i R_i \otimes S_i$ and Schmidt rank $K$. We prove that the conjecture holds for $M$ with $K=3$ and some special $M$ with arbitrary $K$.
\end{abstract}

\maketitle


\section{Introduction}

Given a bipartite state $\r$ acting on the Hilbert space
$\cH_A\ox \cH_B$, the partial transpose w.r.t. an
orthonormal basis $\{\ket{a_i}\}\in\cH_A$ is defined as
$\r^\G=\sum_{ij}\ketbra{a_j}{a_i}\ox\bra{a_i}\r\ket{a_j}$. We say that $\r$ is an $m\times n$ state if $\r_A$ and $\r_B$ respectively have rank $m$ and $n$. If $\r^\G$ is positive semidefinite, then we say that $\r$ is positive partial transpose (PPT).

The partial transpose is a matrix operation with extensive applications in quantum information theory. It is known that if $\r$ is a separable state, i.e., the convex sum of pure product states, then $\r$ is PPT. The converse is still true if $mn\le 6$. It gives the first operational criterion of approaching the separability problem, which is an NP-hard problem. Second, the bipartite PPT entangled states were constructed in 1997 \cite{horodecki1997}. Furthermore the two-qutrit PPTES of rank four was respectively fully characterized in
\cite{Chen2012Equivalence} and \cite{2011Three}. The PPT entangled states are not distillable under local operations and classical communications (LOCC), though
some PPT entangled states may construct distillable key \cite{Horodecki2008Low}. On the other hand, it is conjectured that some non-PPT (NPT) states are not distillable too. This is the long-standing distillability problem \cite{Divincenzo2000Evidence}.
In spite of much efforts devoted in the past decades, not much progress has been made \cite{hh1999, Chen2008Rank,Lin2016Non}. In the papers \cite{Chen2011Multicopy, Chen2012NONDISTILLABLE}, authors  investigated the distillability of three bipartite reduced density operators $\r_{AB},\r_{AC}$ and $\r_{BC}$ from a tripartite pure state $\ket{\ps}_{ABC}$. It indicates that the distillability of one reduced operator is restricted by the other, so it gives a novel criterion of determining the distillability. Naturally, one hopes to develop a similar criterion in terms of tripartite mixed states. However, the problem is hard due to little understanding of the three bipartite reduced density operators of mixed states. In fact, so far we cannot determine the tradeoff between the ranks of these reduced density operators.

In the paper \cite{chl14}, authors proposed a conjectured inequality describing the tradeoff. The inequality claims that $r(\r_{AB})\cdot r(\r_{BC})\ge r(\r_{AC})$ where $r(M)$ denotes the rank of $M$. In the same paper, authors have proven that the conjecture is equivalent to another conjectured inequality, namely
$r(M)\le K \cdot r(M^\G)$ holds for any bipartite matrix $M$ with $K$ being its Schmidt rank. They also showed that the inequality holds for $K=2$. In this paper, we prove that the inequality holds for $M$ with $K=3$, and a special family of matrices with arbitrary positive integer $K$. They are presented in Theorem \ref{thm:sr3} and \ref{thm:special}, respectively. The proof of Theorem \ref{thm:sr3} is based on Lemma \ref{le:2xn,sr=3} and \ref{le:kr=rr3}.
Next using the equivalence between the two inequalities above, we show that $r(\r_{AB})\cdot r(\r_{BC})\ge r(\r_{AC})$ holds when $r(\r_{BC})\le3$ in Corollary \ref{cr:rbc<=rabxrac}. We apply our results to investigate the distillability of bipartite states in Lemma \ref{le:4<=2*2}. That is if $\r_{AB}$ and $\r_{AC}$ have both rank two and $\r_{BC}$ is NPT, then $\r_{BC}$ has rank at most four. Hence $\r_{BC}$ is distillable. Furthermore, if $\r_{AB}$ has rank at most three and $\r_{BC}\in\cB(\bbC^m\otimes\bbC^n)$ is NPT and $\max\{m,n\}\ge r(\r_{AB})\cdot r(\r_{AC})$, then $\r_{BC}$ has rank at most $\max\{m,n\}$. So $\r_{BC}$ is distillable. The results show the tradeoff between the ranks of three bipartite reduced density operators. Further, they determine the distillability of one of the three bipartite reduced density operators by using that of the other two of them.
We thus manage to extend the results on distillability in \cite{Chen2011Multicopy, Chen2012NONDISTILLABLE} from tripartite pure states to mixed states.

It is known that many quantum-information tasks need entangled pure states as necessary resources, while merely mixed states exist because of to the noise in nature. So we need to convert mixed states into pure states under LOCC. This is why entanglement distillation and the distillability problem has been widely studied in theory and experiment for decades \cite{Abdelkhalek2016Efficient,Takeoka2017Unconstrained,Guo2017Entanglement, Rozp2018Optimizing,Dehaene2003Local,Datta2012Compact,Lamata2006Relativity,Dong2008Experimental, Takahashi2010Entanglement}. The problem is related to other quantum-information problems, such as the super-activation of zero capacity quantum channels and bound entanglement. Since mixed states have a more fruitful configuration in physics than pure states, our results thus shows novel understanding of distillability of bipartite states from the viewpoint of tripartite system.

The rest of this paper is organized as follows. In Sec. \ref{sec:pre} we introduce the preliminary knowledge and facts for the proofs in subsequent sections. In Sec. \ref{sec:sr3}, we prove Conjecture \ref{cj:1} for matrices of Schmidt rank three. We apply our results to investigate the distillability of bipartite reduced density matrices from the same tripartite mixed states. We further prove the conjecture for special matrices of arbitrary Schmidt rank in Sec. \ref{sec:special}. Finally we conclude in Sec. \ref{sec:con}.

\section{Preliminaries}
\label{sec:pre}

In this section we introduce the preliminary knowledge used in this paper. We introduce quantum information in Sec. \ref{subsec:qi}, linear algebra  and the main conjecture in Sec. \ref{subsec:linearalg}.

\subsection{quantum information}
\label{subsec:qi}

Let $\bbM_{m,n}$ be the set of $m\times n$ complex matrices, and $\bbM_n:=\bbM_{n,n}$. In quantum physics, we say that the positive semidefinite matrix $\r\in \bbM_m\otimes\bbM_n$ is a  bipartite quantum state on the Hilbert space $\cH_A\ox\cH_B=\bbC^m\ox\bbC^n$ with the normalization condition $\tr\r=1$. The state $\r$ is a pure state when the rank of $\r$ is one. If $\r$ has rank at least one then we say. that $\r$ is a mixed state.
We denote a pure state as $\r=\proj{\ps}$ for some unit vector $\ket{\ps}$, i.e., $\norm{\ps}:=\braket{\ps}{\ps}=1$. The partial trace w.r.t. system A is defined as
$\r_B:=\text{$\tr$}_A \r :=
\sum_i \bra{a_i}_A\r\ket{a_i}_A$,
where $\ket{a_i}$ is an arbitrary orthonormal basis in $\cH_A$. We call $\r_B$ the reduced density operator of system $B$. One can similarly define the reduced density operator of system $A$ as $\r_A:=\tr_B \r$. 

Let $\ket{x,y}:=\ket{x}\otimes\ket{y}$ for any states $\ket{x},\ket{y}$.
It is known that every bipartite pure state $\ket{\ps}$ can be written as the Schmidt decomposition. That is, $\ket{\ps}=\sum_i \sqrt{c_i} \ket{a_i,b_i}$ where $c_i\ge0$, $\sum_i c_i=1$, $\{\ket{a_i}\}$ is an orthonormal basis in $\cH_A$ and $\{\ket{b_i}\}$ is an orthonormal basis in $\cH_B$. One can obtain that the reduced density operators of $\ket{\ps}$ are $\r_A=\tr_B\r=\sum_i c_i \proj{a_i}$, and $\r_B=\tr_A\r=\sum_i c_i \proj{b_i}$. We define the partial transpose w.r.t. system A as $M^{\G}:=M^{\G_A}:=\ketbra{i}{j}_AM\ketbra{i}{j}_A$. One can similarly define $M^{\G_B}$.

\subsection{the conjecture}
\label{subsec:linearalg}

We denote $r(M)$ as the rank of matrix $M$, $M^T$ as the transpose of $M$, $M^*$ as the complex conjugate of $M$, and $M^\dg=(M^*)^T$. Let $\bbM_{m,n}$ be the set of $m\times n$ complex matrices. Any matrix $M\in \bbM_{m_1,n_1}\ox\bbM_{m_2,n_2}$ can be written as $M=\sum^{m_1}_{i=1}\sum^{n_1}_{j=1}\ketbra{i}{j}\otimes M_{i,j}$ with $M_{i,j}\in \bbM_{m_2,n_2}.$ We denote $0_q$ as the $q\times q$ zero matrix. Now we present the main problem as the following conjecture from \cite{chl14}. This is the main problem we investigate in this paper.
\bcj
 \label{cj:1}
Let $R_1,\cdots,R_k$ be $m_1\times n_1$ complex matrices, and let
$S_1,\cdots,S_k$ be $m_2\times n_2$ complex matrices. Then
 \bea
 \label{eq:cj2}
 r \bigg( \sum^K_{i=1} R_i \ox S_i \bigg)
\le
 K\cdot r \bigg( \sum^K_{i=1} R_i \ox S_i^T \bigg).
\eea
If we assume $M=\sum^K_{i=1} R_i \ox S_i$, then the inequality is equivalent to $
r(M) \le K\cdot r(M^\Gamma).	
$
\qed
\ecj

Evidently, it suffices to prove the conjecture when the matrices $R_i$'s are linearly independent, and the matrices $S_i$'s are also linearly independent. In this case, the integer $K:=K(M)$ is the \textit{Schmidt rank} of the matrix $M=\sum^K_{i=1} R_i \otimes S_i$. One can derive that
$K(M)
=
K(M^{\G_A})=K(M^{\G_B})
=
K(M^\dg)=K(M^T)=K(M^*).	
$
If $K(M)\ge\max\{r(M),r(M^\G)\}$ then Conjecture \ref{cj:1} hold. Next it follows from \cite{chl14} that Conjecture \ref{cj:1} holds for $K(M)\le2$. So it suffices to prove Conjecture \ref{cj:1} under the assumption
\begin{eqnarray}
\label{eq:k<=}
&&
2<K(M)<\max\{r(M),r(M^\G)\},
\\&&
\label{eq:2<km<=}
2<K(M)\le\min\{m_1n_1,m_2n_2\}.
\end{eqnarray}
We shall show in Lemma \ref{le:sr=m1n1} that the second inequality in \eqref{eq:2<km<=} is strict.

Next, if we multiply local invertible matrix $V_l\otimes W_l$ on the lhs of $\sum^K_{i=1} R_i \ox S_i$, then the rank is unchanged. One can obtain the similar result when the multiplication is performed on the rhs of $\sum^K_{i=1} R_i \ox S_i$. We shall refer to the multiplication as that the rank of matrix is unchanged up to local equivalence.

More explicitly, we denote locally equivalent $M,N$ as $M\sim N$, namely there exist invertible product matrix $U\otimes V$ and $W\otimes X$ such that $(U\otimes V)M(W\otimes X)=N$. Hence, proving Conjecture \ref{cj:1} is equivalent to proving it up to local equivalence. That is, proving the inequality \eqref{eq:cj2} is equivalent to proving the inequality
$
r \bigg(
(A\otimes B)
\bigg(
\sum^K_{i=1} R_i \ox S_i
\bigg)
 (C\otimes D)\bigg)
\le
K\cdot
r \bigg(
(E\otimes F)
\bigg(
\sum^K_{i=1} R_i \ox S_i^T
\bigg)
 (G\otimes H)\bigg)$,
where the matrices $A,E,L\in\bbM_{m_1}\otimes \bbM_{m_1},...,D,F,M\in\bbM_{n_2}\otimes \bbM_{n_2}$ are invertible matrices we can choose arbitrarily.

To conclude this section we present the following observation as a special case of proving Conjecture \ref{cj:1}.

\begin{lemma}
\label{le:sr=m1n1}
If $\sr(M)=\min\{m_1n_1,m_2n_2\}$, then Conjecture \ref{cj:1} holds for $M$.
\end{lemma}
\begin{proof}
Let $M=[M_{ij}]$, and $k=\max r(M_{ij})$. Then $\sr(M)\cdot r(M)=m_1n_1 \cdot r(M)\ge m_1n_1k\ge r(M^\Gamma)$. The last inequality follows from the fact that every block $M_{ij}$ has rank at most $k$.
\end{proof}

\section{Conjecture \ref{cj:1} for matrices of Schmidt rank three and applications}
\label{sec:sr3}

In this section we prove Conjecture \ref{cj:1} for matrices $M$ of Schmidt rank three. This is the first case satisfying the assumptions in Eqs. \eqref{eq:k<=} and \eqref{eq:2<km<=}. We begin by characterizing a special $M$ in Lemma \ref{le:2xn,sr=3}. Then we prove that Conjecture \ref{cj:1} holds for $M$ in Lemma \ref{le:kr=rr3}. Assisted by this lemma, we present the main result of this section in Theorem \ref{thm:sr3}.
\begin{lemma}
\label{le:2xn,sr=3}
(i) Suppose the block matrix $\bma M_{11} & M_{12}\\ M_{21} & M_{22} \ema$ has Schmidt rank three. Then it is locally equivalent to the matrix 	$N=\bma N_{11} & N_{12}\\ N_{21} & wN_{11} \ema$ where $N_{11}=\bma I_k & 0\\ 0 & 0 \ema$, $\rank N_{11}=k$, and $w$ is a complex number.

(ii) The matrix $N$ is locally equivalent to $\bma N_{11} & w^{-1}N_{12}\\ N_{21} & N_{11} \ema$ or $\bma N_{11} & N_{12}\\ N_{21} & 0 \ema$. Both of them and their partial transpose have Schmidt rank three.

\end{lemma}
\begin{proof}
(i) Let $M=\bma M_{11} & M_{12}\\ M_{21} & M_{22} \ema$. Up to local equivalence we may assume that $M_{11},M_{12},M_{21}$ are linearly independent. Since $\sr(M)=3$, we obtain that $M_{22}$ is the linear combination of $M_{11},M_{12},M_{21}$. Let $M_{22}=xM_{11}+yM_{12}+zM_{21}$. By a block-row operation on $M$, we obtain $M\sim M'=\bma M_{11} & M_{12}\\ M_{21}-yM_{11} & xM_{11}+zM_{21} \ema$. By a block-column operation on $M'$, we obtain that $M'\sim M''=\bma M_{11} & M_{12}-zM_{11}\\ M_{21}-yM_{11} & (x+yz)M_{11} \ema$. Let $\rank N_{11}=k$. Setting $w=x+yz$ and replace $N_{11}$ by $\bma I_k & 0\\ 0 & 0 \ema$ up to local equivalence imply the assertion.

(ii) The assertion follows from the two cases, $w\ne0$ and $w=0$.
\end{proof}

Now we are in a position to prove the main result of this section. We show its proof in Appendix \ref{app:appen}.

\begin{theorem}
\label{thm:sr3}
Conjecture \ref{cj:1} holds for any matrix of Schmidt rank three.
\end{theorem}
We apply the theorem to investigate the distillability of bipartite states. In  \cite{chl14}, it has been proven that Conjecture \ref{cj:1} is equivalent to the inequality
$
r(\r_{AB})\cdot r(\r_{AC})\ge r(\r_{BC}).
$
In particular, if Conjecture \ref{cj:1} holds for the integer $K$ then the inequality holds when $K=r(\r_{AB})$. By results in \cite{chl14} and Theorem \ref{thm:sr3} of this paper, we have

\begin{corollary}
\label{cr:rbc<=rabxrac}
The inequality
$r(\r_{AB})\cdot r(\r_{AC})\ge r(\r_{BC})
$ holds when $r(\r_{AB})\le3$.
\end{corollary}

It implies the following fact.

\begin{lemma}
\label{le:4<=2*2}
Suppose $\r_{ABC}$ is a tripartite state.

(i) If $\r_{AB}$ and $\r_{AC}$ have both rank two and $\r_{BC}$ is NPT, then $\r_{BC}$ is distillable.

(ii) If $\r_{AB}$ has rank at most three and $\r_{BC}\in\cB(\bbC^m\otimes\bbC^n)$ is NPT and $\max\{m,n\}\ge r(\r_{AB})\cdot r(\r_{AC})$, then $\r_{BC}$ is distillable.
\end{lemma}
\begin{proof}
(i) It follows from Corollary \ref{cr:rbc<=rabxrac} that $\r_{BC}$ has rank at most $r_{AB}\cdot r_{AC}=4$. Since $\r_{BC}$ is NPT, it follows from \cite{Lin2016Non} that $\r_{BC}$ is distillable.

(ii) It follows from Corollary \ref{cr:rbc<=rabxrac} that $\r_{BC}$ has rank at most $r(\r_{AB})\cdot r(\r_{AC})\le\max\{m,n\}$. Since $\r_{BC}$ is NPT, it follows from \cite{Chen2012Distillability} that $\r_{BC}$ is distillable.
\end{proof}

The above fact extends the results on distillability in \cite{Chen2011Multicopy, Chen2012NONDISTILLABLE} from tripartite pure states to tripartite mixed states. In particular, Lemma \ref{le:4<=2*2} (i) says that for distillable states $\r_{AB},\r_{AC}$ of rank two we obtain distillable $\r_{BC}$. The similar argument for states in higher dimensions can be obtained by Lemma \ref{le:4<=2*2} (ii).

\section{Conjecture \ref{cj:1} for special matrices of arbitrary Schmidt rank}
\label{sec:special}

In the previous sections, we have shown that Conjecture \ref{cj:1} holds for matrices $M$ of Schmidt rank three. In this section, we prove that Conjecture \ref{cj:1} holds for some special $M$ of arbitrary Schmidt rank. This is presented in the following observation.
\begin{theorem}
\label{thm:special}
Three  special cases of $M$ with Schmidt rank $s$ ($s\leq n_1$) satisfy Conjecture \ref{cj:1}. We discuss as follows:

(i) One row of $M$ has $s$ linearly independent blocks.

(ii) Each row of $M$ has only one linearly independent block.

(iii) One row of $M$ has $s-1$ linearly independent blocks and the remaining blocks are zero. Meanwhile, one of the $s$-th blocks linearly independent with the former $s-1$ linearly independent blocks is below zero blocks.
\end{theorem}
\begin{proof}
Define $M=\sum_{i=1}^s R_i\otimes S_i$, where $R_i$ , $S_i$ are $m_1\times n_1$ and $m_2\times n_2$ complex matrices, respectively, $i=1,2,\cdots,s$. $M$ is the minimal counterexample and $M$ has Schmidt rank $s$, in the sense that $m_1+n_1$ takes the smallest possible value. We write
\begin{eqnarray}
M=
\bma
M_{11}&M_{12}&\cdots&M_{1n_1}\\
M_{21}&M_{22}&\cdots&M_{2n_1}\\
\vdots&\vdots&\ddots&\vdots\\
M_{m_11}&M_{m_12}&\cdots&M_{m_1n_1}
\ema.
\end{eqnarray}
According to Theorem (\ref{thm:sr3}), in the same way, we assume $U:=$ span $\{S_1, S_2,\cdots, S_s\}$, meanwhile, $V:=\mathcal{C}(S_1\quad S_2\quad\cdots\quad S_s)$. So every column of each $M_{ij}$ is in $V$. Suppose that $T_1, T_2,\cdots,T_s$ are linearly independent blocks of $M$. Then we have span $\{T_1, T_2,\cdots,T_s\}=U$ and $\mathcal{C}(S_1\quad S_2\quad\cdots\quad S_s)=\mathcal{C}(T_1\quad T_2\quad\cdots \quad T_s)=V$. In particular, we have
\begin{eqnarray}
\dim V\leq \rank T_1+\rank T_2+\cdots+\rank T_s,
\end{eqnarray}
where at least one $\rank T_i\geq\frac{1}{s}\dim V$. There are three cases as follows:

(i) Suppose that $M_{11}, M_{12}, \cdots, M_{1s}$ are linearly independent blocks, then span$\{M_{11}, M_{12},\cdots, M_{1s}\}=U$. By applying block-wise column operations we assume
\begin{eqnarray}
M=
\bma
M_{11}&M_{12}&\cdots&M_{1s}&0&\cdots&0\\
M_{21}&M_{22}&\cdots&M_{2s}&M_{2 s+1}&\cdots&M_{2n_1}\\
\vdots&\vdots&\vdots&\vdots&\vdots&\ddots&\vdots\\
M_{m_11}&M_{m_12}&\cdots&M_{m_1s}&M_{m_1s+1}&\cdots&M_{m_1n_1}
\ema,
\end{eqnarray}
and
\begin{eqnarray}
M'=
\bma
M_{2 s+1}&\cdots&M_{2n_1}\\
\vdots&\ddots&\vdots\\
M_{m_1s+1}&\cdots&M_{m_1n_1}
\ema.
\end{eqnarray}
Let $M_{11},M_{12},\cdots,M_{1s}$ be $T_1, T_2,\cdots, T_s$. Then by applying the first inequality of (\ref{eq:73}), we obtain
\begin{eqnarray}
\rank M
&\geq&
\rank (M_{11}\quad M_{12}\quad\cdots\quad M_{1s})+\rank M'
\\
&=&
\dim V+\rank M'.
\end{eqnarray}
Besides, because each row of $M^{\Gamma_B}$ is in the space $V\oplus V\oplus\cdots\oplus V$ (s-copies of V), then by applying the second inequality of (\ref{eq:73}), we have
\begin{eqnarray}
\rank M^{\Gamma_B}\leq s\dim V+\rank (M')^{\Gamma_B}.
\end{eqnarray}
If $\rank M^{\Gamma_B}\geq s\rank M$, then we have $\rank (M')^{\Gamma_B}\geq s\rank M'$, which violates what we assume.

(ii) Suppose that each row has only one linearly independent block. Because at least one $\rank T_i\geq\frac{1}{s}\dim V$, we may assume $\rank T_1\geq \frac{1}{s}\dim V$. By applying block-wise column operations, we assume that the matrix has the form
\begin{eqnarray}
M=
\bma
T_1&0&\cdots&0\\
M_{21}&M_{22}&\cdots&M_{2n_1}\\
\vdots&\vdots&\ddots&\vdots\\
M_{m_11}&M_{m_12}&\cdots&M_{m_1n_1}
\ema,
\end{eqnarray}
then each row of $M^{\Gamma_B}$ is in the space $V$. So we obtain
\begin{eqnarray}
\rank M^{\Gamma_B}\leq\dim V+\rank (M')^{\Gamma_B},
\end{eqnarray}
where $M'$ is the part of the matrix below zero blocks.
According to the first inequality of (\ref{eq:73}), we have
\begin{eqnarray}
\rank M\geq\frac{1}{s}\dim V+\rank M'.
\end{eqnarray}
So if $\rank M^{\Gamma_B}>s\cdot \rank M$, we obtain $\rank (M')^{\Gamma_B}>s\cdot \rank M'$, which violates that the $M$ is the minimal counterexample.

(iii) Suppose that one row has $s-1$ linearly independent blocks, the remaining blocks are zero. Meanwhile, one of blocks linearly independent with $M_{11}, M_{12}, \cdots, M_{1s-1}$ is below the zero blocks. Without loss of generality, we may assume that $M_{11}, M_{12}, \cdots, M_{1s-1}$ are linearly independent blocks. By applying block-wise row and column operations,  we may assume that $M_{2s}$ is the block linearly independent with $M_{11}, M_{12}, \cdots, M_{1s-1}$. So we
obtain
\begin{eqnarray}
\label{eq:78}
M=
\bma
M_{11}&M_{12}&\cdots&M_{1s-1}&0&\cdots&0\\
M_{21}&M_{22}&\cdots&M_{2s-1}&M_{2 s}&\cdots&M_{2n_1}\\
\vdots&\vdots&\vdots&\vdots&\vdots&\ddots&\vdots\\
M_{m_11}&M_{m_12}&\cdots&M_{m_1s-1}&M_{m_1s}&\cdots&M_{m_1n_1}
\ema.
\end{eqnarray}
By applying block-wise column operations, then we have $M_{2j}$ are combination of $M_{11}, M_{12}, \cdots, M_{1s-1}$, $1\leq j\leq s-1$. We assume $M_{2j}=a_{1j}M_{11}+a_{2j}M_{12}+\cdots+a_{s-1j}M_{1s-1}$, then multiply an appropriate constant $k$ to the first row of (\ref{eq:78}) and add to the second row such that
\begin{eqnarray}
\label{eq:88}
\det
\left|\begin{array}{ccccc}
k+a_{11}&a_{21}&\cdots&a_{s-11}&0\\
a_{12}&k+a_{22}&\cdots&a_{s-12}&0\\
\vdots&\vdots&\ddots&\vdots&\vdots\\
a_{1s-2}&a_{2s-2}&\cdots&k+a_{s-1s-2}&0\\
a_{1s-1}&a_{2s-1}&\cdots&a_{s-1s-1}&1
\end{array}\right|\neq0.
\end{eqnarray}
By block-wise row operations, we obtain that the first row have $s$ linearly independent blocks, which case has been solved in (i).
\end{proof}

In spite of Theorem \ref{thm:sr3} and \ref{thm:special}, the proof for Conjecture \ref{cj:1} with arbitrary $M$ remains an open problem. We expect that the idea of the previous two sections be applied to Conjecture \ref{cj:1} with $M$ of Schmidt rank four.

\section{Conclusions}
\label{sec:con}

We have shown that Conjecture \ref{cj:1} holds for matrix $M$ with $K=3$, and some special $M$ with arbitrary $K$. We have  applied our results to determine the distillability of quantum entanglement three bipartite reduced states from the same tripartite mixed states. Our result extends the results in terms of tripartite pure states in \cite{Chen2011Multicopy, Chen2012NONDISTILLABLE}. Our results provide general technique for understanding and proving Conjecture \ref{cj:1} for arbitrary $M$. This is also the next mission in this direction. Another direction is to find more extension in terms of reduction criterion, PPT, tripartite states with a qubit. We may also explore the distillability of bipartite reduced density operators of multipartite states.

\section*{Acknowledgments}
	\label{sec:ack}	
Authors were supported by the  NNSF of China (Grant No. 11871089), and the Fundamental Research Funds for the Central Universities (Grant Nos. KG12080401 and ZG216S1902).



\appendix

\section{Appendix: Proof of Theorem \ref{thm:sr3}}
\label{app:appen}

We begin by proving a special case of Theorem \ref{thm:sr3}, namely when the block matrix satisfies $m_1=n_1=2$.

\begin{lemma}
\label{le:kr=rr3}
Conjecture \ref{cj:1} holds for the block matrix $\bma M_{11} & M_{12}\\ M_{21} & M_{22} \ema$ of Schmidt rank three.
\end{lemma}
\begin{proof}
Let $M=\bma M_{11} & M_{12}\\ M_{21} & M_{22} \ema$. It suffices to prove the assertion by assuming that each block $M_{ij}$ is a $d\times d$ matrix. From Lemma \ref{le:2xn,sr=3}, $M$ is locally equivalent to the matrix $N=\bma N_{11} & w^{-1}N_{12}\\ N_{21} & N_{11} \ema$ or $\bma N_{11} & N_{12}\\ N_{21} & 0 \ema$, where $N_{11}=\bma I_k & 0\\ 0 & 0 \ema$, $\rank N_{11}=k$, and $w$ is a nonzero complex number. Using the equivalence, we discuss two cases (i) and (ii).

(i) Suppose $N=\bma N_{11} & w^{-1}N_{12}\\ N_{21} & N_{11} \ema$. Let
$w^{-1}N_{12}:=
\bma
Q_{11}&Q_{12}\\
Q_{21}&Q_{22}
\ema$,
and $N_{21}:=
\bma
R_{11}&R_{12}\\
R_{21}&R_{22}
\ema$, where $Q_{11}$ and $R_{11}$ are $k\times k$ blocks. We have
\begin{eqnarray}
\label{eq:ikk1}
M\sim N=
\bma
I_k&0&Q_{11}&Q_{12}\\
0&0&Q_{21}&Q_{22}\\
R_{11}&R_{12}&I_k&0\\
R_{21}&R_{22}&0&0
\ema.
\end{eqnarray}
By block-row operations on $N$, we have
\begin{eqnarray}
\label{eq:nsimn1}
N\sim N_1=
\bma
I_k&0&Q_{11}&Q_{12}\\
0&0&Q_{21}&Q_{22}\\
0&R_{12}&-R_{11}Q_{11}+I_k&-R_{11}Q_{12}\\
0&R_{22}&-R_{21}Q_{11}&-R_{21}Q_{12}
\ema.
\end{eqnarray}
Note that $\rank M=\rank N=\rank N_1$, $\sr(M)=3$, and the Schmidt rank of $N_1$ may be not three. One will see that it does not influence the subsequent argument.

Let $\rank(\bma Q_{21}&Q_{22}\ema)=r_4$ and $\rank(\bma R_{12}\\ R_{22}\ema)=r_1$. Using \eqref{eq:nsimn1} we obtain
\begin{eqnarray}
\label{eq:r1r4}
\rank M=\rank N_1 \geq k+r_1+r_4.
\end{eqnarray}
In a similar way to \eqref{eq:ikk1}, we have
\begin{eqnarray}
M\sim N_2=
\bma
-Q_{11}R_{11}+I_k&-Q_{11}R_{12}&0&Q_{12}\\
-Q_{21}R_{11}&-Q_{21}R_{12}&0&Q_{22}\\
R_{11}&R_{12}&I_k&0\\
R_{21}&R_{22}&0&0
\ema.
\end{eqnarray}
Let $\rank (\bma R_{21}&R_{22}\ema)=r_2$ and $\rank (\bma Q_{12}\\Q_{22}\ema)=r_3$. So we obtain
\begin{eqnarray}
\label{eq:r2r3}
\rank M=\rank N_2\geq k+r_2+r_3.
\end{eqnarray}
Thus, from (\ref{eq:r1r4}) and (\ref{eq:r2r3}), we have
\begin{eqnarray}
\label{eq:mgeq12}
\rank M\geq\frac{1}{2}(r_1+r_2+r_3+r_4)+k.
\end{eqnarray}
Using  \eqref{eq:ikk1}, we have
\begin{eqnarray}
M^{\Gamma_A}=
\bma
I_k&0&R_{11}&R_{12}\\
0&0&R_{21}&R_{22}\\
Q_{11}&Q_{12}&I_k&0\\
Q_{21}&Q_{22}&0&0
\ema.
\end{eqnarray}
Then we obtain
\begin{eqnarray}
\label{eq:mgammaleq}
\rank M^{\Gamma_A}\leq 2k+r_1+r_3.
\end{eqnarray}
Using \eqref{eq:mgeq12} and \eqref{eq:mgammaleq}, we have $\rank M^{\Gamma_A}\leq3\rank M$. So Conjecture \ref{cj:1} holds for (i).

(ii) Suppose $N=\bma N_{11} & N_{12}\\ N_{21} & 0 \ema$. Similarly, let
$N_{12}:=
\bma
Q_{11}&Q_{12}\\
Q_{21}&Q_{22}
\ema$,
and $N_{21}:=
\bma
R_{11}&R_{12}\\
R_{21}&R_{22}
\ema$, where $Q_{11}$ and $R_{11}$ are $k\times k$ blocks. We have
\begin{eqnarray}
\label{eq:ikk2}
M\sim N=
\bma
I_k&0&Q_{11}&Q_{12}\\
0&0&Q_{21}&Q_{22}\\
R_{11}&R_{12}&0&0\\
R_{21}&R_{22}&0&0
\ema.
\end{eqnarray}
By block-wise row and column operations on $N$, we have
\begin{eqnarray}
N\sim N_1=
\bma
I_k&0&0&0\\
0&0&Q_{21}&Q_{22}\\
0&R_{12}&-Q_{11}R_{11}&-Q_{12}R_{11}\\
0&R_{22}&-Q_{11}R_{21}&-Q_{12}R_{21}
\ema.
\end{eqnarray}
Let $\rank(\bma Q_{21}&Q_{22}\ema)=r_4$ and $\rank(\bma R_{12}\\R_{22}\ema)=r_1$. So we obtain
\begin{eqnarray}
\label{eq:ikk5}
\rank M=\rank N_1\geq k+r_1+r_4.
\end{eqnarray}
Using \eqref{eq:ikk2}, we have
\begin{eqnarray}
M^{\Gamma_A}=
\bma
I_k&0&R_{11}&R_{12}\\
0&0&R_{21}&R_{22}\\
Q_{11}&Q_{12}&0&0\\
Q_{21}&Q_{22}&0&0
\ema.
\end{eqnarray}
Then we obtain
\begin{eqnarray}
\label{eq:ikk6}
\rank M^{\Gamma_A}
&\leq&
k+\rank(\bma Q_{11}&Q_{12}\ema)+r_4\\
&+&
\rank (\bma R_{11}\\R_{21}\ema)+r_1
\\
&\leq&
3k+r_1+r_4.
\end{eqnarray}
Using (\ref{eq:ikk5}) and (\ref{eq:ikk6}), we have $\rank M^{\Gamma_A}\leq 3\rank M$.
So the Conjecture \ref{cj:1} holds for (ii).
\end{proof}

Now we show the proof of Theorem \ref{thm:sr3}.

\begin{proof}
Define
\begin{eqnarray}
M=R_1\otimes S_1+R_2\otimes S_2+R_3\otimes S_3,
\end{eqnarray}
where $R_i$ , $S_i$ are $m_1\times n_1$ and $m_2\times n_2$ complex matrices, respectively, i=1,2,3. $M$ has Schmidt rank three. In the sense that $m_1+n_1$ takes the smallest possible value. i.e., $M$ is the minimal counterexample. We write
\begin{eqnarray}
\label{eq:mt66}
M=
\bma
M_{11}&M_{12}&\cdots&M_{1n_1}\\
M_{21}&M_{22}&\cdots&M_{2n_1}\\
\vdots&\vdots&\ddots&\vdots\\
M_{m_11}&M_{m_12}&\cdots&M_{m_1n_1}
\ema,
\end{eqnarray}
where $M_{ij}=(R_1)_{ij}S_1+(R_2)_{ij}S_2+(R_3)_{ij}S_3$. Then we assume $U:=$ span $\{S_1, S_2, S_3\}$, meanwhile, $V:=\mathcal{C}(S_1\quad S_2\quad S_3)$, where $\mathcal{C}(A)$ denotes the columns span of matrix $A$. So every column of each $M_{ij}$ is in $V$.
Suppose that $T_1, T_2, T_3$ are linearly independent blocks of $M$. Then we have span $\{T_1, T_2, T_3\}=U$. By applying elementary column operations we obtain
\begin{eqnarray}
\mathcal{C}(T_1\quad T_2\quad T_3)\nonumber
&=&
\mathcal{C}(T_1\quad T_2\quad T_3\quad0\quad0\quad0)\\\nonumber
&=&
\mathcal{C}(T_1\quad T_2\quad T_3\quad S_1\quad S_2\quad S_3)\\\nonumber
&=&
\mathcal{C}(0\quad 0\quad 0\quad S_1\quad S_2\quad S_3)\\\nonumber
&=&
\mathcal{C}(S_1\quad S_2\quad S_3)\\
&=&
V.
\end{eqnarray}
Then consider the following inequalities, which hold for arbitrary block matrices:
\begin{eqnarray}
\label{eq:73}
\rank A+\rank C
&\leq& \rank \bma A&0\\B&C\ema\nonumber\\
&\leq&
 \rank \bma A\\B\ema +\rank C.
\end{eqnarray}
In particular, we have
\begin{eqnarray}
\label{eq:7m}
\dim V\leq\rank T_1+\rank T_2+\rank T_3,
\end{eqnarray}
hence $\rank T_i\geq\frac{1}{3}\dim V$ for at least one value of $i$. Because if all $\rank T_i<\frac{1}{3}\dim V$, we have
\begin{eqnarray}
\rank T_1+\rank T_2+\rank T_3<\dim V,
\end{eqnarray}
which violates (\ref{eq:73}). There are three cases as follows:

(i) Suppose that $M_{11}, M_{12}, M_{13}$ are linearly independent matrices, then span$\{M_{11}, M_{12}, M_{13}\}=U$. By applying block-wise column operations we assume
\begin{eqnarray}
M=
\bma
M_{11}&M_{12}&M_{13}&0&\cdots&0\\
M_{21}&M_{22}&M_{23}&M_{24}&\cdots&M_{2n_1}\\
\vdots&\vdots&\vdots&\vdots&\ddots&\vdots\\
M_{m_11}&M_{m_12}&M_{m_13}&M_{m_14}&\cdots&M_{m_1n_1}
\ema,
\end{eqnarray}
and
\begin{eqnarray}
M'=
\bma
M_{24}&\cdots&M_{2n_1}\\
\vdots&\ddots&\vdots\\
M_{m_14}&\cdots&M_{m_1n_1}
\ema.
\end{eqnarray}
Then we assume that $M_{11}, M_{12}, M_{13}$ are $T_1, T_2, T_3$. By applying the first inequality of (\ref{eq:73}), we obtain
\begin{eqnarray}
\rank M
&\geq&
\rank (M_{11}\quad M_{12}\quad M_{13})+\rank M'\\
&=&
\dim V+\rank M'.
\end{eqnarray}
Besides, each row of $M^{\Gamma_B}$ is in the space $V\oplus V\oplus V$, and apply the second inequality of (\ref{eq:73}), then we have
\begin{eqnarray}
\label{eq:79}
\rank
M^{\Gamma_B}
&\leq&
\rank
\bma
M_{11}^T&M_{12}^T&M_{13}^T\\
M_{21}^T&M_{22}^T&M_{23}^T\\
\vdots&\vdots&\vdots\\
M_{m_11}^T&M_{m_12}^T&M_{m_13}^T
\ema
+
\rank
(M')^{\Gamma_B}\\
&\leq&
3\dim V+\rank (M')^{\Gamma_B}.
\end{eqnarray}
If $\rank M^{\Gamma_B}>3\rank M$, we have
\begin{eqnarray}
3\dim V+3\rank M'
&\leq&
 3\rank M< \rank M^{\Gamma_B}\nonumber\\
 &\leq&
 3\dim V+\rank (M')^{\Gamma_B},
\end{eqnarray}
then $M'$ is a counterexample. But we assume that the $M$ is the minimal counterexample, we have $\rank (M')^{\Gamma_B}>3\rank M'$. Then it violates what we assume.

(ii) Suppose that each row has only one linearly independent block. Because at least $\rank T_i\geq\frac{1}{3}\dim V$, we may assume $\rank T_1\geq \frac{1}{3}\dim V$. By applying block-wise column operations, we assume that the matrix has the form
\begin{eqnarray}
M=
\bma
T_1&0&\cdots&0\\
M_{21}&M_{22}&\cdots&M_{2n_1}\\
\vdots&\vdots&\ddots&\vdots\\
M_{m_11}&M_{m_12}&\cdots&M_{m_1n_1}
\ema,
\end{eqnarray}
and
\begin{eqnarray}
M''=
\bma
M_{22}&\cdots&M_{2n_1}\\
\vdots&\ddots&\vdots\\
M_{m_12}&\cdots&M_{m_1n_1}
\ema.
\end{eqnarray}
Then each row of $M^{\Gamma_B}$ is in the space $V$. So we obtain
\begin{eqnarray}
\rank M^{\Gamma_B}\leq\dim V+\rank (M'')^{\Gamma_B},
\end{eqnarray}
where $M''$ is the part of the matrix below zero blocks.
According to the first inequality of (\ref{eq:73}), we have
\begin{eqnarray}
\rank M
&\geq&
 \rank T_1+\rank M''\\
&\geq&
\frac{1}{3}\dim V+\rank M''.
\end{eqnarray}
So if $\rank M^{\Gamma_B}>\rank M$, we obtain
\begin{eqnarray}
\dim V+3\rank M''
&\leq&
 \rank M<\rank M^{\Gamma_B}\nonumber\\
 &\leq&
 \dim V+\rank (M'')^{\Gamma_B}.
\end{eqnarray}
Thus, we have $\rank (M'')^{\Gamma_B}>3\rank M''$, which violates that the $M$ is the minimal counterexample.

(iii) Suppose that each row has at most two linearly independent blocks. Without loss of generality, let $M_{11}, M_{12}$ be $T_{1}, T_2$, respectively. We have
\begin{eqnarray}
M=
\bma
M_{11}&M_{12}&0&\cdots&0\\
M_{21}&M_{22}&M_{23}&\cdots&M_{2n_1}\\
\vdots&\vdots&\vdots&\ddots&\vdots\\
M_{m_11}&M_{m_12}&M_{m_13}&\cdots&M_{m_1n_1}
\ema.
\end{eqnarray}
Then we assume that $M_{23}$ is $T_3$. By applying block-wise column and row operations, we obtain that the first row has three linearly independent blocks, which case has been solved in (i).

Actually, if $T_3$ is not in the part of the matrix below zero blocks , we may assume $T_3=M_{21}$, then let
\begin{eqnarray}
\label{eq:mij}
M_{ij}=x_{ij}M_{11}+y_{ij}M_{12},\quad i\geq2, j\geq3.
\end{eqnarray}
According to Lemma \ref{le:kr=rr3}, assume $M_{22}=wM_{11}$. Then we have
\begin{eqnarray}
M\sim M_1=
\bma
M_{11}&M_{12}&0&0&\cdots&0\\
M_{21}&wM_{11}&M_{23}&M_{24}&\cdots&M_{2n_1}\\
\vdots&\vdots&\vdots&\vdots&\ddots&\vdots\\
M_{m_11}&M_{m_12}&M_{m_13}&M_{m_14}&\cdots&M_{m_1n_1}
\ema.
\end{eqnarray}

If one of $M_{i1}$ is linearly independent with $M_{11}$, $M_{21}$ in $M_1$, we have the first row of $M^T$ has three linearly independent blocks, which case has been solved in (i). So by applying block-wise row operations, we have
\begin{eqnarray}
M_1\sim M_2=
\bma
M_{11}&M_{12}&0&0&\cdots&0\\
M_{21}&wM_{11}&M_{23}&M_{24}&\cdots&M_{2n_1}\\
0&M_{32}&M_{33}&M_{34}&\cdots&M_{3n_1}\\
\vdots&\vdots&\vdots&\vdots&\ddots&\vdots\\
0&M_{m_12}&M_{m_13}&M_{m_14}&\cdots&M_{m_1n_1}
\ema.
\end{eqnarray}

According to whether $M_{ij}$ are zero blocks, $i\geq2, j\geq3$, we have two cases (iii.a), (iii.b). In (iii.a), we assume that $M_{ij}$ are not all zero blocks, $i\geq2, j\geq3$. Then in (iii.b), we assume that $M_{ij}$ are all zero blocks, $i\geq2, j\geq3$. We discuss (iii.a) and (iii.b) as follows:

(iii.a) Without loss of generality, if $M_{ij}$ are not all zero blocks, $i\geq2, j\geq3$, we may assume that $M_{23}$ is nonzero. Because when one of $M_{2j}, j\geq3$ is nonzero block, say $M_{2s}$.  By applying block-wise column operations, $M_{2s}$ as new $M_{23}$ is nonzero block. Besides, when $M_{2j}, j\geq3$ is zero block, we obtain that at least one of $M_{ij}$ is nonzero, $i,j\geq3$. We may assume that $M_{sj}$ is nonzero, then add the second row to the $s$-th row as the new second row. By block-wise column operations, then we have $M_{sj}$ as new $M_{23}$ is nonzero. On the other hand, we obtain that $M_{2j}, j\geq4$ is linearly dependent with $M_{23}$. Otherwise, we have three linearly independent blocks. Then by block-wise column operations, we have
\begin{eqnarray}
\label{eq:53}
M_2\sim M_3=
\bma
M_{11}&M_{12}&0&0&\cdots&0\\
M_{21}&wM_{11}&M_{23}&0&\cdots&0\\
0&M_{32}&M_{33}&M_{34}&\cdots&M_{3n_1}\\
\vdots&\vdots&\vdots&\vdots&\ddots&\vdots\\
0&M_{m_12}&M_{m_13}&M_{m_14}&\cdots&M_{m_1n_1}
\ema.
\end{eqnarray}
According to (\ref{eq:53}), there exist two cases:

When $w\neq0$, if $x_{23}=0, y_{23}\neq0$ or $x_{23}\neq0, y_{23}=0$, add the second row to the first row, then the first row has three linearly independent blocks. If $x_{23}\neq0$ and $y_{23}\neq0$, then $M_{21}, wM_{11}, x_{23}M_{11}+y_{23}M_{12}$ are three linearly independent blocks. These cases are all solved in (i).

When $w=0$, if $x_{23}\neq0$, add the second row to the first row, then the first row has three linearly independent blocks. If $x_{23}=0, y_{23}\neq0$, we have
\begin{eqnarray}
\label{eq:54}
M_3=
\bma
M_{11}&M_{12}&0&0&\cdots&0\\
M_{21}&0&y_{23}M_{12}&0&\cdots&0\\
0&M_{32}&M_{33}&M_{34}&\cdots&M_{3n_1}\\
\vdots&\vdots&\vdots&\vdots&\ddots&\vdots\\
0&M_{m_12}&M_{m_13}&M_{m_14}&\cdots&M_{m_1n_1}
\ema.
\end{eqnarray}

Without loss of generality, we add the first column of (\ref{eq:54}) to the fourth column, then we have $M_{i4}, 3\leq i\leq m_1$ is linearly dependent with $M_{11}$. Otherwise, $M_{11}$, $M_{21}$ and $M_{i4}$ are three linearly independent blocks. In the same way, we obtain that $M_{ij}, 3\leq i\leq m_1, 3\leq j\leq n_1$ is linearly dependent with $M_{11}$. By block-wise row and column operations, we may assume $M_{34}=x_{34}M_{11}, x_{34}\neq0$. Then add to the second row, $M_{21}$, $M_{33}+y_{23}M_{12}$ and $M_{34}$ are three linearly independent blocks. So we obtain that $M_{ij}, 3\leq i\leq m_1, 4\leq j\leq n_1$ are all zero. By block-wise row operations, then we have
\begin{eqnarray}
\label{eq:56}
M_3\sim M_4=
\bma
M_{11}&M_{12}&0&0&\cdots&0\\
M_{21}&0&y_{23}M_{12}&0&\cdots&0\\
0&M_{32}&x_{33}M_{11}&0&\cdots&0\\
0&M_{42}&0&0&\cdots&0\\
\vdots&\vdots&\vdots&\vdots&\ddots&\vdots\\
0&M_{m_12}&0&0&\cdots&0
\ema.
\end{eqnarray}

From (\ref{eq:56}), add the first column to the second column, then we have $M_{i2}, i\geq3$ is linear combination of $M_{11}+M_{12}$ and $M_{21}$. We may assume $M_{i2}=a_{i2}(M_{11}+M_{12})+b_{i2}M_{21}, i\geq3$. If one of $a_{i2}$ is nonzero, say $M_{s2}, 3\leq s\leq m_1$, then multiply an appropriate constant $k$ to the first column of (\ref{eq:56}) and add to the second column such that
\begin{eqnarray}
\label{eq:57}
\det
\left|\begin{array}{ccc}
a_{s2}&a_{s2}&b_{s2}\\
k&1&0\\
0&0&k
\end{array}\right|=k(1-k)a_{s2}\neq0.
\end{eqnarray}
From (\ref{eq:57}), we obtain that $kM_{11}+M_{12}, kM_{21}, M_{s2}$ are three linearly independent blocks.

Besides, if $a_{i2}=0$ and one of $M_{i2}$ is nonzero, say $M_{s2}$. We multiply an appropriate constant $k$ to the first row, then add the second and third row to the first row such that $kM_{11}+M_{21}, kM_{12}+M_{s2}, y_{23}M_{12}+x_{33}M_{11}$ are three linearly independent blocks. When $M_{i2}$ is zero, we have
\begin{eqnarray}
\label{eq:58}
M_4=
\bma
M_{11}&M_{12}&0&0&\cdots&0\\
M_{21}&0&y_{23}M_{12}&0&\cdots&0\\
0&0&x_{33}M_{11}&0&\cdots&0\\
0&0&0&0&\cdots&0\\
\vdots&\vdots&\vdots&\vdots&\ddots&\vdots\\
0&0&0&0&\cdots&0
\ema.
\end{eqnarray}
According to (\ref{eq:58}), there are two cases. If $x_{33}\neq0$, by applying block-wise row operations, we have $M_{11}+M_{21}, M_{12}$ and $M_{12}+x_{33}M_{11}$ are three linearly independent blocks. If $x_{33}=0$, we have
\begin{eqnarray}
\label{eq:59}
M=
\bma
M_{11}&M_{12}&0&0&\cdots&0\\
M_{21}&0&y_{23}M_{12}&0&\cdots&0\\
0&0&0&0&\cdots&0\\
0&0&0&0&\cdots&0\\
\vdots&\vdots&\vdots&\vdots&\ddots&\vdots\\
0&0&0&0&\cdots&0
\ema,
\end{eqnarray}
where $y_{23}\neq0$.
Thus, we obtain
\begin{eqnarray}
\rank M
&\leq&
\rank M_{11}+\rank M_{21}+\rank \bma M_{12}&0\\0&y_{23}M_{12}\ema\nonumber\\
&\leq&
3\rank M^\Gamma.
\end{eqnarray}
In the same way, we have $\rank M^\G\leq3\rank M$. So we have finished the proof of (iii.a).

(iii.b) We consider  $M_{ij}, i\geq2, j\geq3$ are all zero blocks. By applying block-wise row operations, we have
\begin{eqnarray}
\label{eq:85}
M=
\bma
M_{11}&M_{12}&0&\cdots&0\\
M_{21}&wM_{11}&0&\cdots&0\\
\vdots&\vdots&\vdots&\ddots&\vdots\\
M_{m_11}&M_{m_12}&0&\cdots&0
\ema.
\end{eqnarray}
From (\ref{eq:85}), we obtain that $M_{i1}$ is linear combination of $M_{11}$ and $M_{21}$, $i\geq3$. Otherwise, the first row of $M^T$ has three independent blocks. Then we have
\begin{eqnarray}
\label{eq:63}
M=
\bma
M_{11}&M_{12}&0&\cdots&0\\
M_{21}&wM_{11}&0&\cdots&0\\
0&M_{32}&0&\cdots&0\\
\vdots&\vdots&\vdots&\ddots&\vdots\\
0&M_{m_12}&0&\cdots&0
\ema.
\end{eqnarray}

If $M_{i2}$ is zero, $i\geq3$, we have
\begin{eqnarray}
M=
\bma
M_{11}&M_{12}&0&\cdots&0\\
M_{21}&wM_{11}&0&\cdots&0\\
0&0&0&\cdots&0\\
\vdots&\vdots&\vdots&\ddots&\vdots\\
0&0&0&\cdots&0
\ema,
\end{eqnarray}
this case is solved by Lemma \ref{le:kr=rr3}.

If one of $M_{i2}$ is nonzero, $i\geq3$, say $M_{s2}$. Then $M_{s2}$ is linear combination of $M_{12}+M_{11}$ and $M_{21}+wM_{11}$. We assume $M_{s2}=a_3M_{11}+b_3M_{12}+c_3M_{21}$ and at least one of $a_3, b_3, c_3$ is not zero. Then multiply an appropriate constant $k$ to the first column of (\ref{eq:63}) and add to the second column such that
\begin{eqnarray}
\label{eq:65}
\det
\left|\begin{array}{ccc}
k&1&0\\
w&0&k\\
a_3&b_3&c_3
\end{array}\right|=-b_3k^2-wc_3+a_3k\neq0.
\end{eqnarray}
If at least one of $b_3, wc_3$ and $a_3$ is not zero, then (\ref{eq:65}) holds. We consider $b_3, wc_3$ and  $a_3$ are all zero. Because $M_{s2}$ is nonzero, we have $c_3$ is nonzero. So we have
\begin{eqnarray}
M=
\bma
M_{11}&M_{12}&0&\cdots&0\\
M_{21}&0&0&\cdots&0\\
0&c_3M_{21}&0&\cdots&0\\
0&0&0&\cdots&0\\
\vdots&\vdots&\vdots&\ddots&\vdots\\
0&0&0&\cdots&0
\ema,
\end{eqnarray}
which case is solved by (\ref{eq:59}). So we have finished the proof of (iii.b).

Thus, we obtain that Conjecture \ref{cj:1} holds for the matrix of Schmidt rank three.

\end{proof}

\bibliographystyle{unsrt}

\bibliography{Entanglementdistillationintermsofaconjecturedmatrixinequality}

\begin{thebibliography}{10}

\bibitem{horodecki1997}
P.~Horodecki.
\newblock Separability criterion and inseparable mixed states with positive
  partial transposition.
\newblock {\em Phys. Lett. A}, 232:333, 1997.

\bibitem{Chen2012Equivalence}
Lin Chen and Dragomir~Ž Đoković.
\newblock Equivalence classes and canonical forms for two-qutrit entangled
  states of rank four having positive partial transpose.
\newblock {\em Journal of Mathematical Physics}, 53(10):805--813, 2012.

\bibitem{2011Three}
?ukasz Skowronek.
\newblock Three-by-three bound entanglement with general unextendible product
  bases.
\newblock {\em Journal of Mathematical Physics}, 52(12):722--725, 2011.

\bibitem{Horodecki2008Low}
K.~Horodecki, L.~Pankowski, M.~Horodecki, and P.~Horodecki.
\newblock Low dimensional bound entanglement with one-way distillable
  cryptographic key.
\newblock {\em IEEE Transactions on Information Theory}, 54(6):2621--2625,
  2008.

\bibitem{Divincenzo2000Evidence}
David~P. Divincenzo, Peter~W. Shor, John~A. Smolin, Barbara~M. Terhal, and
  Ashish~V. Thapliyal.
\newblock Evidence for bound entangled states with negative partial transpose.
\newblock {\em Physical Review A}, 61(6):200--200, 2000.

\bibitem{hh1999}
M.~Horodecki and P.~Horodecki.
\newblock Reduction criterion of separability and limits for a class of
  distillation protocols.
\newblock {\em Phys. Rev. A}, 59:4206, 1999.

\bibitem{Chen2008Rank}
Lin Chen and Yi~Xin Chen.
\newblock Rank three bipartite entangled states are distillable.
\newblock {\em Physical Review A}, 78(2):3674--3690, 2008.

\bibitem{Lin2016Non}
Chen Lin and Dragomir~Z Djokovic.
\newblock Non-positive-partial-transpose quantum states of rank four are
  distillable.
\newblock 2016.

\bibitem{Chen2011Multicopy}
Lin Chen and Masahito Hayashi.
\newblock Multicopy and stochastic transformation of multipartite pure states.
\newblock {\em Physical Review A}, 83(2):4795--4804, 2011.

\bibitem{Chen2012NONDISTILLABLE}
Lin Chen and Masahito Hayashi.
\newblock Nondistillable entanglement guarantees distillable entanglement.
\newblock {\em International Journal of Modern Physics B}, 26(27n28):1243008--,
  2012.

\bibitem{chl14}
Josh Cadney, Marcus Huber, Noah Linden, and Andreas Winter.
\newblock Inequalities for the ranks of multipartite quantum states.
\newblock {\em Linear Algebra and its Applications}, 452(0):153 -- 171, 2014.

\bibitem{Abdelkhalek2016Efficient}
Daniela Abdelkhalek, Mareike Syllwasschy, Nicolas~J. Cerf, Jaromír Fiurá?ek,
  and Roman Schnabel.
\newblock Efficient entanglement distillation without quantum memory.
\newblock {\em Nature Communications}, 7:11720, 2016.

\bibitem{Takeoka2017Unconstrained}
Masahiro Takeoka, Kaushik~P. Seshadreesan, and Mark~M. Wilde.
\newblock Unconstrained capacities of quantum key distribution and entanglement
  distillation for pure-loss bosonic broadcast channels.
\newblock {\em Phys.rev.lett}, 119(15):150501, 2017.

\bibitem{Guo2017Entanglement}
Ying Guo, Cailang Xie, Qin Liao, Wei Zhao, Guihua Zeng, and Duan Huang.
\newblock Entanglement-distillation attack on continuous-variable quantum key
  distribution in a turbulent atmospheric channel.
\newblock {\em Phys.rev.a}, 96(2):022320, 2017.

\bibitem{Rozp2018Optimizing}
Rozp?dek Filip, Schiet Thomas, Thinh~Le Phuc, Elkouss David, Doherty~Andrew C.,
  and Wehner Stephanie.
\newblock Optimizing practical entanglement distillation.
\newblock {\em Physical Review A}, 97(6):062333--, 2018.

\bibitem{Dehaene2003Local}
Jeroen Dehaene, Maarten Van~den Nest, Bart De~Moor, and Frank Verstraete.
\newblock Local permutations of products of bell states and entanglement
  distillation.
\newblock {\em Physical Review A}, 67(2):426--430, 2003.

\bibitem{Datta2012Compact}
Animesh Datta, Lijian Zhang, Joshua Nunn, Nathan~K Langford, and Ian~A
  Walmsley.
\newblock Compact continuous-variable entanglement distillation.
\newblock {\em Physical Review Letters}, 108(6):060502, 2012.

\bibitem{Lamata2006Relativity}
L.~Lamata, M.~A. Martin-Delgado, and E.~Solano.
\newblock Relativity and lorentz invariance of entanglement distillability.
\newblock {\em Physical Review Letters}, 97(25):250502, 2006.

\bibitem{Dong2008Experimental}
Ruifang Dong, Mikael Lassen, Joel Heersink, Christoph Marquardt, Radim Filip,
  Gerd Leuchs, and Ulrik~L. Andersen.
\newblock Experimental entanglement distillation of mesoscopic quantum states.
\newblock {\em Nature Physics}, 4(12):919--923, 2008.

\bibitem{Takahashi2010Entanglement}
Hiroki Takahashi, Jonas~S. Neergaard-Nielsen, Makoto Takeuchi, Masahiro
  Takeoka, Kazuhiro Hayasaka, Akira Furusawa, and Masahide Sasaki.
\newblock Entanglement distillation from gaussian input states.
\newblock {\em Nature Photonics}, 4(3):178--181, 2010.

\bibitem{Chen2012Distillability}
Lin Chen and Dragomir~Z Djokovic.
\newblock Distillability and ppt entanglement of low-rank quantum states.
\newblock {\em Journal of Physics A Mathematical and Theoretical},
  44(28):1213--1219, 2012.

\end{thebibliography}

\end{document}